\documentclass[10pt]{article}
\usepackage[a4paper,margin=1in]{geometry}
\usepackage[T1]{fontenc}
\usepackage[utf8]{inputenc}
\usepackage{amsfonts,latexsym,amsmath,amsthm,amssymb}
\usepackage{newtxtext,newtxmath}
\usepackage[numbers,sort&compress]{natbib}
\usepackage[colorlinks,citecolor=black,urlcolor=black,linkcolor=black]{hyperref}
\hypersetup{
  pdftitle={Exact Risk-Complexity Laws for Projective Boundaries in Scenario Optimization and Distribution-Free Certification},
  pdfauthor={Giuseppe C. Calafiore},
  pdfkeywords={scenario optimization, distribution-free certification, conformal prediction, finite-sample risk, sample compression}
}
\usepackage{graphicx}
\usepackage{xcolor}
\usepackage{booktabs}
\usepackage{enumitem}
\usepackage{microtype}
\graphicspath{{./}{figures/}}
\allowdisplaybreaks
\setlist{nosep}

\theoremstyle{plain}
\newtheorem{theorem}{Theorem}[section]
\newtheorem{lemma}[theorem]{Lemma}
\newtheorem{proposition}[theorem]{Proposition}
\newtheorem{corollary}[theorem]{Corollary}

\theoremstyle{definition}
\newtheorem{definition}[theorem]{Definition}
\newtheorem{assumption}[theorem]{Assumption}
\newtheorem{example}[theorem]{Example}

\theoremstyle{remark}
\newtheorem{remark}[theorem]{Remark}

\newcommand{\Pp}{\mathbb P}
\newcommand{\Ee}{\mathbb E}
\newcommand{\one}{\mathbf 1}
\newcommand{\cZ}{\mathcal Z}

\newcommand{\cP}{\mathcal P}
\newcommand{\Rr}{\mathbb R}
\newcommand{\eps}{\varepsilon}
\DeclareMathOperator{\argmin}{argmin}
\DeclareMathOperator{\Beta}{Beta}

\title{Exact Risk-Complexity Laws for Projective Boundaries in Scenario Optimization and Distribution-Free Certification}
\author{Giuseppe C. Calafiore\\
\small Department of Electronics and Telecommunications, Politecnico di Torino, Italy\\
\small Corresponding author: \href{mailto:giuseppe.calafiore@polito.it}{giuseppe.calafiore@polito.it}}
\date{}

\begin{document}
\maketitle

\begin{abstract}
Scenario optimization, conformal prediction, and related distribution-free certification methods use finite samples to construct decisions or prediction sets with violation-risk guarantees for fresh observations. In several classical settings, the conditional violation risk follows an exact beta law, whose tail has a beta-binomial representation and whose parameter is a support, calibration, or compression dimension. This paper identifies the deterministic boundary mechanism behind these formulas and derives the corresponding law when the observed boundary size is random. A decision rule is represented by an acceptance set for future observations, together with a boundary map selecting the sample points responsible for that set. The resulting pair is called a {\em proper projective boundary scheme} when held-out samples are accepted precisely if the full-sample boundary is retained, and accepted non-boundary samples can be deleted without changing that boundary. For every such scheme, the conditional law of the violation risk given the observed boundary size is determined by the boundary's cross-sample complexity profile. A stable profile yields the usual beta law, whereas a varying profile produces an exact profile correction. The framework covers scalar order-statistic calibration, support-reconstructive scenario programs, cascaded support-removal certificates, coordinatewise envelopes, and Pareto-frontier calibration with vector scores. It also yields conditional probabilistic certificates and a no-go result explaining why observed complexity alone is insufficient.
\end{abstract}

\noindent\textbf{Keywords:} scenario optimization; distribution-free certification; conformal prediction; finite-sample risk; sample compression.

\medskip
\noindent\textbf{Mathematics Subject Classification (2020):} Primary 90C15, 62G15; Secondary 90C90, 68Q32.

\section{Introduction}\label{sec:introduction}

This paper studies exact finite-sample risk-complexity laws for randomized decision rules and distribution-free certification methods.  Scenario optimization and conformal prediction are two motivating examples.  In both settings, a finite sample is used to construct an acceptance set for a future observation, and the central performance quantity is the conditional probability that a fresh observation is not accepted.  The aim is to identify when this risk has an exact beta law and what replaces that law when the effective boundary size is random.

Scenario optimization and conformal prediction use samples in different ways but often lead to similar finite-sample formulas.  In scenario optimization, sampled constraints replace an uncertain constraint, and one studies the probability that the optimizer violates a fresh constraint; see \cite{CalafioreCampi2005,CalafioreCampi2006,CampiGaratti2008,CalafioreSIOPT2010,CampiGaratti2011,CampiGaratti2018,GarattiCampi2022,GarattiCampi2025}.  In conformal prediction, calibration data are used to build prediction sets with finite-sample coverage under weak distributional assumptions; see \cite{PapadopoulosEtAl2002,VovkGammermanShafer2005,ShaferVovk2008,LeiEtAl2018,AngelopoulosBates2023}.  Recent work also studies multivariate conformal prediction, risk control, and links between conformal and scenario methods; see \cite{DheurEtAl2025,JohnstoneNdiaye2025,TawachiLauferGoldshtein2025,OSullivanRomaoMargellos2026}.  In both areas, beta-binomial expressions occur naturally.

For example, in a nondegenerate convex scenario program with \(N\) sampled constraints and deterministic support size \(s\), the violation probability \(V_N\) satisfies
\begin{equation}\label{eq:intro-scenario-tail}
  \Pp\{V_N>\eps\}
  =
  \sum_{i=0}^{s-1}\binom{N}{i}\eps^i(1-\eps)^{N-i}.
\end{equation}
A scalar split-conformal predictor with a continuous nonconformity score has a risk law with the same beta form.  If the threshold is the \(s\)-th largest calibration score, then the conditional miscoverage probability has distribution \(\Beta(s,N-s+1)\), which is the standard order-statistic law; see \cite{DavidNagaraja2003}.

We argue that a deterministic boundary property of the sample is the common mechanism behind these formulas, rather than convexity, scalar scoring, or the observed support size alone.  Informally, a new observation is rejected exactly when it would become one of the observations that determines the decision after augmentation.  The term ``boundary'' is used here as a common name for support constraints and essential sets in scenario optimization, for compression sets in learning theory, and for the calibration observations that determine a split-conformal quantile.  The formal concept is given in Section~\ref{sec:proper-boundaries}.

The size of the boundary set is a key complexity indicator, and many useful procedures have random boundary size.  One example is conformal prediction with vector-valued scores.  Suppose a calibration example has a score in \(\Rr_+^d\), where the coordinates represent different residuals, losses, or safety margins.  A Pareto-frontier rule accepts a candidate if its vector score is componentwise no worse than at least one calibration score.  The boundary is then the empirical set of nondominated calibration scores, and its size is random.  Conditioning on the observed frontier size by itself does not give a beta law.

The risk-complexity theory of \cite{GarattiCampi2022,GarattiCampi2023Conditional,CampiGaratti2023JMLR} already shows that the observed complexity carries information about out-of-sample risk, and also that conditional risk statements based only on the observed complexity require additional information.  The present paper gives a boundary-level version of that message.  It identifies the exact profile information needed to pass from an observed boundary size to a conditional risk law.

The main object is a decision rule \(\Gamma\) that maps a finite sample to an acceptance set.  A future point is accepted if it belongs to this set, and is a violation otherwise.  To the decision rule, we associate a boundary map \(B_n\), which selects the indices of the samples that determine the decision.  The pair \((\Gamma,B)\) is a \emph{proper projective boundary scheme} if it satisfies two deterministic conditions: all held-out samples are accepted if and only if the full boundary is retained, and once the boundary is retained, removing accepted non-boundary samples does not change the boundary.  These conditions are stated for abstract schemes and cover convex and nonconvex optimization settings, together with applications outside optimization.

Under these assumptions, if
\[
  V_N=\Pp\{Z\notin \Gamma(S_N)\mid S_N\},\qquad
  K_N=|B_N(S_N)|,\qquad
  p_n(k)=\Pp\{K_n=k\},
\]
then for every \(M\ge0\),
\begin{equation}\label{eq:intro-main}
  \Ee[(1-V_N)^M\mid K_N=k]
  =
  \frac{p_{N+M}(k)}{p_N(k)}
  \frac{\binom{N}{k}}{\binom{N+M}{k}},
\end{equation}
whenever \(p_N(k)>0\).  The conditional law of \(V_N\) is therefore determined by the complexity profile \(\{p_n(k)\}_{n\ge k}\).  If the profile is stable at the observed value, the profile factor disappears and \(V_N\mid\{K_N=k\}\sim\Beta(k,N-k+1)\).  If the profile changes with the sample size, the beta law is generally wrong after conditioning on \(K_N=k\).

The consequences are useful in both scenario optimization and distribution-free prediction.  First, the result gives a diagnostic for deciding when familiar beta-binomial certificates are exact: the relevant boundary must be proper and projective, and its cross-sample profile must be stable at the observed complexity.  Second, it gives the exact profile-corrected law for random-boundary procedures, rather than treating a random frontier or support size as a fixed dimension.  Third, it turns conditional risk certification into a finite-sample problem of computing, bounding, or estimating the complexity profile.  Fourth, the no-go result in Section~\ref{sec:no-go} shows that this extra profile information is necessary for nontrivial distribution-free conditional guarantees.

The paper has four main aims.  It states boundary equivalence and projectivity in a form that can be checked directly.  It proves the exact moment law \eqref{eq:intro-main} and the corresponding profile-based conditional certificate.  It verifies the assumptions for scalar order-statistic calibration, support-reconstructive scenario programs, a coordinatewise random-support envelope, and Pareto-frontier calibration with vector scores.  It also clarifies the role of discarded samples: essential projective discards may enter the sharp law, while generic violated-discard procedures call for specialized scenario-discarding bounds or conservative compression bounds; see \cite{CalafioreSIOPT2010,CampiGaratti2011,RomaoPapachristodoulouMargellos2023,RomaoMargellosPapachristodoulou2023}.

The paper is therefore a finite-sample risk-complexity result for abstract decision rules, with scenario optimization as a central optimization instance and conformal prediction as a parallel distribution-free instance.  The profile \(p_n(k)\) plays the role of a complexity law for the calibration or decision rule.  In elementary schemes it is analytic; in structured schemes it can often be bounded by deterministic arguments; and in simulator-access settings it can be estimated with simultaneous finite-sample bands.  These routes are made explicit in Section~\ref{subsec:profile-computation}.

Longer proofs and verifications are collected in the appendices.  Appendix~\ref{app:padded-proof} proves the bounded-boundary result.  Appendix~\ref{app:cascaded-support} treats cascaded support removal, while Appendices~\ref{app:primitive} and~\ref{app:inner} give a reusable verification lemma and the inner-certificate result.

\section{Setup}\label{sec:setup}

Let \((\cZ,\mathcal A)\) be the measurable space of one observation.  In a
scenario program, \(z\) is typically an uncertainty realization; in conformal
prediction, \(z\) may be a labelled example, a residual, or a calibration score.
Throughout the single-risk part of the paper, \(Z_1,Z_2,\ldots\) are i.i.d.
with common law \(P\), and \(S_n=(Z_1,\ldots,Z_n)\).  For deterministic data we
write \(T_n=(z_1,\ldots,z_n)\).  If \(I\subseteq[n]=\{1,\ldots,n\}\), then
\(T_I\) denotes the corresponding subcollection, written in a fixed deterministic
order.  All rules are assumed to be permutation invariant; this ordering is only
a notational device.

A decision rule is a measurable set-valued map \(\Gamma\) that sends a finite
data set \(T_I\) to an acceptance set
$
  \Gamma(T_I)\in\mathcal A$.
Equivalently, the indicator \(\one\{z\in\Gamma(T_I)\}\) is jointly measurable in
\((T_I,z)\).  A fresh observation \(z\) is accepted when
\(z\in\Gamma(T_I)\) and is a violation when \(z\notin\Gamma(T_I)\).  At sample
size \(N\), the conditional violation risk is
\begin{equation}\label{eq:risk}
  V_N:=\Pp\{Z\notin\Gamma(S_N)\mid S_N\},
\end{equation}
where \(Z\sim P\) is independent of \(S_N\).

In a scenario program, \(\Gamma(S_N)\) may be
\[
  \{\delta:g(x^*(S_N),\delta)\le0\},
\]
where \(x^*(S_N)\) is the optimizer returned by the sampled problem and \(g\) is
the constraint function.  Then \(V_N\) is the usual violation probability.  In
split conformal prediction, \(\Gamma(S_N)\) is a set of future examples, labels,
or score vectors accepted by the calibration rule.  Then \(V_N\) is the
conditional miscoverage probability.

We use the following standing conventions.  Ties are resolved by a fixed
measurable tie-breaker that is equivariant under permutations, or else excluded
by a non-atomicity assumption.  Boundary maps are assumed measurable, so events
such as \(\{K_n=k\}\) are well defined.  These regularity assumptions are standard in scenario and conformal arguments.  A construction may also be verified on a measurable regularity class of probability one, provided that the identities hold simultaneously for all subcollections used in the exchangeability argument.  Since only finitely many subcollections occur at each sample size, the probabilistic conclusions are unchanged.

\section{Proper Projective Boundaries}\label{sec:proper-boundaries}

A \emph{boundary} is the part of the data that is responsible for the decision.
The definition below is deterministic and valid for every finite sample
size in the range where the scheme is used; the random results later come only
from applying the deterministic statement to i.i.d. data.

\begin{definition}[Boundary map and complexity]
A boundary map is a permutation-e\-qui\-va\-riant rule that assigns to every finite
data set \(T_n\) a subset
$
  B_n(T_n)\subseteq[n]$.
The boundary complexity is
$
  K_n(T_n):=|B_n(T_n)|$.
For random samples, write \(K_n:=K_n(S_n)\) and
\[
  p_n(k):=\Pp\{K_n=k\}.
\]
The sequence \(\{p_n(k):n\ge k\}\) is the complexity profile at level \(k\).
\end{definition}

\begin{assumption}[Boundary equivalence]\label{ass:boundary-equivalence}
For every deterministic data set \(T_n=(z_1,\ldots,z_n)\) and every split
\(I\subseteq[n]\), \(J=[n]\setminus I\),
\begin{equation}\label{eq:boundary-equivalence}
  z_j\in\Gamma(T_I)\ \text{for all }j\in J
  \quad\Longleftrightarrow\quad
  B_n(T_n)\subseteq I.
\end{equation}
\end{assumption}

The left side says that the points left out of the design set all pass the
decision trained on the design set.  The right side says that none of the
left-out points is needed in the full-sample boundary.  Thus a held-out
violation is exactly a point that would enter the full boundary.

\begin{assumption}[Boundary projectivity]\label{ass:boundary-projectivity}
For every deterministic data set \(T_n\) and every \(I\subseteq[n]\), if
\(B_n(T_n)\subseteq I\), then
\begin{equation}\label{eq:boundary-projectivity}
  B_{|I|}(T_I)=B_n(T_n)
\end{equation}
after the natural re-indexing from \(T_I\) back to \(T_n\).
\end{assumption}

Projectivity says that, once all boundary samples have been kept, deleting
accepted non-boundary samples does not change the reported boundary.  This
rules out artificial complexities that depend on irrelevant accepted samples.

\begin{definition}[Proper projective boundary scheme]
A pair \((\Gamma,B)\) satisfying Assumptions
\ref{ass:boundary-equivalence} and \ref{ass:boundary-projectivity} is called a {\em
proper projective boundary scheme}.
\end{definition}

\begin{remark}[Discarded samples]
A discarded sample can be part of \(B_n\), but only if it is essential for the
certified decision.  For example, in scalar conformal prediction the discarded
upper-tail scores and the threshold-defining score form an order-statistic
boundary.  In a scenario program with a deterministic removal path, the removed
constraints may be boundary samples if they are needed to reconstruct that path
or the certified acceptance set.  A constraint is not a boundary point merely
because it was removed or because the final optimizer violates it.
\end{remark}

\section{The Exact Single-Risk Law}\label{sec:single-risk-law}

The following theorem is the main result of the paper.  It expresses the risk-complexity principle at the level of projective boundaries: the conditional risk law is controlled by how the observed boundary complexity changes when new samples are added.

\begin{theorem}[Exact projective-boundary law]\label{thm:exact-profile}
Assume that \((\Gamma,B)\) is a proper projective boundary scheme and that the
observations are i.i.d.  Fix \(N\ge1\), \(M\ge0\), and
\(k\in\{0,\ldots,N\}\).  Then
\begin{equation}\label{eq:unconditional-moment}
  \Ee\!\left[(1-V_N)^M\one_{\{K_N=k\}}\right]
  =
  \frac{\binom{N}{k}}{\binom{N+M}{k}}\,p_{N+M}(k).
\end{equation}
Consequently, if \(p_N(k)>0\), then
\begin{equation}\label{eq:conditional-moment}
  \Ee\!\left[(1-V_N)^M\mid K_N=k\right]
  =
  \frac{p_{N+M}(k)}{p_N(k)}
  \frac{\binom{N}{k}}{\binom{N+M}{k}}.
\end{equation}
\end{theorem}

\begin{proof}
Let \(Y_1,\ldots,Y_M\) be fresh i.i.d. samples from \(P\), independent of
\(S_N\).  Conditional on \(S_N\),
\[
  \Pp\{Y_1,\ldots,Y_M\in\Gamma(S_N)\mid S_N\}=(1-V_N)^M.
\]
Hence
\begin{equation}\label{eq:moment-probability}
  \Ee[(1-V_N)^M\one_{\{K_N=k\}}]
  =
  \Pp\{Y_1,\ldots,Y_M\in\Gamma(S_N),\ K_N(S_N)=k\}.
\end{equation}

Now form the augmented sample
\[
  T_{N+M}=(Z_1,\ldots,Z_N,Y_1,\ldots,Y_M).
\]
By exchangeability of the augmented sample, the probability in
\eqref{eq:moment-probability} is the same as the following experiment: draw
\(T_{N+M}\), choose uniformly an \(N\)-point design subset
\(I\subseteq[N+M]\), let \(J=[N+M]\setminus I\), and ask for
\[
  T_J\subseteq\Gamma(T_I)
  \quad\text{and}\quad
  K_N(T_I)=k .
\]
By boundary equivalence,
\[
  T_J\subseteq\Gamma(T_I)
  \quad\Longleftrightarrow\quad
  B_{N+M}(T_{N+M})\subseteq I.
\]
On this event, projectivity gives
\[
  B_N(T_I)=B_{N+M}(T_{N+M}),
\]
after re-indexing, and therefore
\(K_N(T_I)=K_{N+M}(T_{N+M})\).  The event is thus equivalent to
\[
  B_{N+M}(T_{N+M})\subseteq I
  \quad\text{and}\quad
  K_{N+M}(T_{N+M})=k.
\]
Conditional on \(T_{N+M}\) and \(K_{N+M}=k\), the boundary is a fixed
\(k\)-element subset of \([N+M]\).  A uniform \(N\)-element subset \(I\)
contains it with probability
\[
  \frac{\binom{N+M-k}{N-k}}{\binom{N+M}{N}}
  =
  \frac{\binom{N}{k}}{\binom{N+M}{k}}.
\]
Taking expectations gives \eqref{eq:unconditional-moment}.  Dividing by
\(p_N(k)\) gives \eqref{eq:conditional-moment}.
\end{proof}

\begin{corollary}[Conditional law from the profile]\label{cor:law-from-profile}
If \(p_N(k)>0\), then the conditional law of \(V_N\) given \(K_N=k\) is the
unique probability measure \(\mu_{N,k}\) on \([0,1]\) satisfying
\begin{equation}\label{eq:profile-moments}
  \int_0^1(1-v)^M\,d\mu_{N,k}(v)
  =
  \frac{p_{N+M}(k)}{p_N(k)}
  \frac{\binom{N}{k}}{\binom{N+M}{k}},
  \qquad M=0,1,2,\ldots.
\end{equation}
\end{corollary}

\begin{proof}
The moments are those of Theorem~\ref{thm:exact-profile}.  By the Hausdorff moment theorem, probability measures on the compact interval \([0,1]\) are determined by their integer moments.
\end{proof}

\begin{remark}[Admissible profiles]
A true boundary scheme automatically produces a valid Hausdorff moment sequence
in \eqref{eq:profile-moments}.  A proposed model or estimate of a complexity
profile must satisfy the same positivity and complete-monotonicity constraints
before it can be used as an exact law.
\end{remark}

\subsection{When the Beta Law Is Valid}\label{subsec:beta-law}

The beta law follows when the observed value is \(K_N=k\) and the complexity profile is stable at that value.

\begin{corollary}[Profile-stable beta law]\label{cor:profile-stable-beta}
Assume the conditions of Theorem~\ref{thm:exact-profile}.  Fix
\(k\in\{0,\ldots,N\}\) with \(p_N(k)>0\).  If
\begin{equation}\label{eq:profile-stability}
  p_{N+M}(k)=p_N(k),\qquad M=0,1,2,\ldots,
\end{equation}
then, for \(k\ge1\),
\[
  V_N\mid\{K_N=k\}\sim \Beta(k,N-k+1).
\]
Equivalently,
\begin{equation}\label{eq:beta-tail}
  \Pp\{V_N>\eps\mid K_N=k\}
  =
  \sum_{i=0}^{k-1}\binom{N}{i}\eps^i(1-\eps)^{N-i},
  \qquad 0\le\eps\le1.
\end{equation}
For \(k=0\), \(V_N=0\) almost surely on \(\{K_N=0\}\).
\end{corollary}

\begin{proof}
Under \eqref{eq:profile-stability},
\[
  \Ee[(1-V_N)^M\mid K_N=k]
  =
  \frac{\binom{N}{k}}{\binom{N+M}{k}}.
\]
If \(U\sim\Beta(k,N-k+1)\), then \(1-U\sim\Beta(N-k+1,k)\) and has the same
moments.  Moment determinacy gives the beta distribution.  The tail expression
is the standard beta-binomial identity.
\end{proof}

\begin{corollary}[Fixed boundary size]\label{cor:fixed-size}
Let \(s\ge0\) be an integer and assume the conditions of Theorem~\ref{thm:exact-profile}.  If \(K_n=s\) almost surely for every \(n\ge\max\{1,s\}\), then, for every \(N\ge\max\{1,s\}\),
\[
  V_N\sim\Beta(s,N-s+1)
\]
when \(s\ge1\), and \(V_N=0\) almost surely when \(s=0\).
\end{corollary}

\begin{proof}
For every \(n\ge\max\{1,s\}\), the assumption gives \(p_n(s)=1\).  Hence the profile is stable at \(s\), and Corollary~\ref{cor:profile-stable-beta} gives the stated law.  The case \(s=0\) follows from the last statement of that corollary.
\end{proof}

\begin{corollary}[Bounded boundary size]\label{cor:bounded_boundary}
Let \(s\ge0\) be an integer and assume the conditions of Theorem~\ref{thm:exact-profile}.  If \(K_n\le s\) almost surely for every \(n\ge\max\{1,s\}\), then, for every \(N\ge\max\{1,s\}\),
\[
  \Pp\{V_N>\eps\}
  \le
  \sum_{i=0}^{s-1}\binom{N}{i}\eps^i(1-\eps)^{N-i},
  \qquad 0<\eps<1.
\]
For \(s=0\), the sum is interpreted as zero.
\end{corollary}

The proof is given in Appendix~\ref{app:padded-proof}.  It augments every observation with an independent auxiliary mark, uses those marks to pad the boundary to cardinality \(s\), and tightens the acceptance set so that the padded scheme remains proper and projective.  The original violation risk is then bounded pointwise by the padded risk, whose fixed-boundary law is \(\Beta(s,N-s+1)\).  This extends the familiar convex scenario bound based on an upper bound for the number of support constraints; see, e.g., \cite{CampiGaratti2008,CalafioreSIOPT2010}.

\section{Verifying the Boundary Assumptions}\label{sec:verification}

This section verifies the boundary assumptions in four representative cases.

\subsection{Scalar Order-Statistic Calibration}\label{subsec:order-statistic}

Let \(\phi:\cZ\to\Rr\) be a measurable score.  Fix \(s\ge1\).  For a finite
design set \(T_I\), define
\[
  \Gamma_s(T_I)=\{z:\phi(z)\le q_s(T_I)\},
\]
where \(q_s(T_I)\) is the \(s\)-th largest score among
\(\{\phi(z_i):i\in I\}\) if \(|I|\ge s\).  If \(|I|<s\), set
\(\Gamma_s(T_I)=\emptyset\).  For a full sample \(T_n\), let \(B_n(T_n)\) be
the indices of the \(s\) largest scores.  Assume scores are distinct, or use a
fixed deterministic tie-breaker.

\begin{proposition}[Order-statistic boundary]\label{prop:order-statistic}
For every \(n\ge s\), \((\Gamma_s,B)\) is a proper projective boundary scheme
and \(K_n=s\).
\end{proposition}

\begin{proof}
Fix \(T_n\) and a split \(I,J\).  If \(B_n(T_n)\subseteq I\), then the top
\(s\) scores in \(T_I\) are the top \(s\) scores in \(T_n\).  Every omitted
non-boundary point has score below \(q_s(T_I)\), so every omitted point is
accepted.

Conversely, suppose \(b\in B_n(T_n)\cap J\).  Since \(b\) is one of the top
\(s\) full-sample scores and is missing from \(I\), the \(s\)-th largest score
in \(I\) is strictly smaller than \(\phi(z_b)\).  Thus \(z_b\notin\Gamma_s(T_I)\).
Boundary equivalence follows.  If \(B_n(T_n)\subseteq I\), the top \(s\) scores
of \(T_I\) and \(T_n\) are the same, which proves projectivity.
\end{proof}

With \(s=r+1\), this is scalar split conformal prediction after discarding
\(r\) upper-tail scores.  The boundary consists of the \(r\) discarded scores
and the threshold-defining score, and Corollary~\ref{cor:fixed-size} gives
\[
  V_N\sim\Beta(r+1,N-r).
\]

\subsection{Support-Reconstructive Scenario Programs}\label{subsec:scenario}

Consider the scenario program
\begin{equation}\label{eq:scenario-program}
  x_I\in\argmin_{x\in X} f(x)
  \quad\text{subject to}\quad g(x,z_i)\le0,\qquad i\in I.
\end{equation}
Assume feasibility for all finite samples under consideration, and that a deterministic
selection rule, for example uniqueness, lexicographic ordering, or a fixed
regularization, is used consistently across all subproblems.  The associated
acceptance set is
$
  \Gamma(T_I)=\{z:g(x_I,z)\le0\}$.
The boundary must be defined with respect to the certified object that is to be
reconstructed.  If the optimizer itself is the certified object, the selected
optimizer must be unique in the above deterministic sense.  If two optimizers
can induce the same acceptance set, then the certified object should instead be
the acceptance set.  Formally, let \(a_I\) denote the certified object and assume
that equality of certified objects is equivalent to equality of the acceptance
sets used for certification.  Define the decision support set
\begin{equation}\label{eq:scenario-support}
  B_n(T_n)=\{i\in[n]:a_{[n]\setminus\{i\}}\ne a_{[n]}\}.
\end{equation}
Assume the support set reconstructs the certified object:
\begin{equation}\label{eq:scenario-reconstruction}
  a_{B_n(T_n)}=a_{[n]}.
\end{equation}
Also assume confirmed-addition stability: if \(I\subseteq K\) and the object
\(a_I\) accepts every added sample in \(K\setminus I\), then \(a_K=a_I\).  For a
convex scenario problem with a unique selected optimizer this is the usual
monotonicity argument: once the old optimizer remains feasible after adding
constraints, no point in the smaller feasible set can improve on it.

\begin{proposition}[Scenario boundary]\label{prop:scenario-boundary}
On any deterministic class of samples for which confirmed-addition stability and
reconstruction \eqref{eq:scenario-reconstruction} hold, the decision support map
\eqref{eq:scenario-support} is a proper projective boundary for \(\Gamma\).
\end{proposition}

\begin{proof}
Fix \(T_n\) and write \(B=B_n(T_n)\).  Suppose first that all samples in
\(J=[n]\setminus I\) are accepted by \(\Gamma(T_I)\).  Confirmed-addition
stability gives \(a_{[n]}=a_I\).  For any \(j\in J\), adding the points in
\(J\setminus\{j\}\) to \(I\) also preserves the certified object, so
\(a_{[n]\setminus\{j\}}=a_{[n]}\).  By \eqref{eq:scenario-support},
\(j\notin B\).  Hence \(B\subseteq I\).

Conversely, suppose \(B\subseteq I\).  By reconstruction, \(a_B=a_{[n]}\).
Every sample outside \(B\) is feasible for \(a_{[n]}\), hence accepted by
\(\Gamma(T_B)\).  Repeated use of confirmed-addition stability gives
\(a_I=a_{[n]}\), and all omitted points are accepted because the full-sample
scenario solution is feasible for every sampled constraint.  This proves
boundary equivalence.

For projectivity, let \(B\subseteq I\).  The preceding paragraph gives
\(a_I=a_{[n]}\).  If \(\ell\in I\setminus B\), then
\(B\subseteq I\setminus\{\ell\}\).  By reconstruction, \(a_B=a_{[n]}\).
Every sample in \((I\setminus\{\ell\})\setminus B\) is feasible for
\(a_{[n]}\), hence accepted by the object reconstructed from \(B\).  Confirmed-addition
stability, applied from \(B\) to \(I\setminus\{\ell\}\), gives
\(a_{I\setminus\{\ell\}}=a_B=a_{[n]}=a_I\).  Thus no index in
\(I\setminus B\) is support in the restricted problem.  Conversely, if
\(j\in B\) and \(a_{I\setminus\{j\}}=a_I\), then every sample in
\([n]\setminus I\) is accepted by \(a_I=a_{[n]}\).  Confirmed-addition
stability would then give \(a_{[n]\setminus\{j\}}=a_{[n]}\), contradicting
\(j\in B\).  Thus the restricted support set is exactly \(B\).
\end{proof}

When the support size is deterministic, Proposition~\ref{prop:scenario-boundary}
and Corollary~\ref{cor:fixed-size} recover the exact scenario law
\eqref{eq:intro-scenario-tail}.  When the support size is random, the
conditional law is governed by the profile ratio in Theorem~\ref{thm:exact-profile},
as in the risk-complexity theory.

\subsection{A Coordinatewise Scenario Envelope with Random Support}\label{subsec:coordinatewise-envelope}

The following elementary scenario problem is useful because its boundary size is
random but the profile is explicit.  Let \(Z=(Z^{(1)},Z^{(2)})\in[0,1]^2\), and
consider
\[
  \min_{x\in\mathbb R^2} x_1+x_2
  \quad\text{subject to}\quad
  Z_i^{(1)}\le x_1,
  \quad Z_i^{(2)}\le x_2,
  \qquad i\in I.
\]
Equivalently, the scalar constraint is
\(g(x,z)=\max\{z^{(1)}-x_1,z^{(2)}-x_2\}\le0\).  For \(I=\emptyset\), set \(\Gamma(T_I)=\emptyset\).  For nonempty \(I\), the optimizer is the
coordinatewise envelope
\[
  x_I=\bigl(\max_{i\in I} Z_i^{(1)},\max_{i\in I} Z_i^{(2)}\bigr),
  \qquad
  \Gamma(T_I)=\{z:z^{(1)}\le x_{I,1},\ z^{(2)}\le x_{I,2}\}.
\]
With continuous marginals, the boundary is the union of the two coordinatewise
maximizers.  Thus \(K_n\in\{1,2\}\): \(K_n=1\) when the same sample maximizes both
coordinates, and \(K_n=2\) otherwise.  The map is proper and projective by the
same argument as Proposition~\ref{prop:scenario-boundary}.

If the two coordinates are independent and continuous, the ranks of the two
coordinatewise maxima are independent and uniform over \([n]\).  Hence
\begin{equation}\label{eq:coordinate-envelope-profile}
  p_n(1)=\frac1n,
  \qquad
  p_n(2)=1-\frac1n.
\end{equation}
For \(N\ge2\), Theorem~\ref{thm:exact-profile} gives, for both \(k=1\) and
\(k=2\),
\begin{equation}\label{eq:coordinate-envelope-moments}
  \mathbb E[(1-V_N)^M\mid K_N=k]
  =\frac{N^2}{(N+M)^2},
  \qquad M=0,1,2,\ldots.
\end{equation}
The moment sequence in \eqref{eq:coordinate-envelope-moments} corresponds to a product of two independent beta variables, rather than to either fixed boundary dimension \(1\) or \(2\):
\[
 \{ 1-V_N\mid K_N=k\}
  \stackrel{d}{=}
  U_1U_2,
  \qquad
  U_1,U_2\overset{\mathrm{i.i.d.}}{\sim}\mathrm{Beta}(N,1),
  \qquad k=1,2.
\]
Equivalently,
\[
  \{-\log(1-V_N)\mid K_N=k\}
  \sim \mathrm{Gamma}(2,\mathrm{rate}\ N).
\]
Thus \(\{V_N\mid K_N=k\}\) has density
\[
  f_{V_N\mid K_N=k}(v)
  =
  N^2[-\log(1-v)](1-v)^{N-1},
  \qquad 0<v<1,
\]
and CDF
\[
  F_{V_N\mid K_N=k}(v)
  =
  1-(1-v)^N\{1-N\log(1-v)\}.
\]
The distribution is the same for \(k=1\) and \(k=2\), because the event that the
two coordinatewise maximizers coincide depends only on the two argmax indices,
which are independent of the two coordinatewise maximum values.

\subsection{Pareto-Frontier Calibration with Vector Scores}\label{subsec:pareto}

Now let scores be vectors.  Write \(u\preceq v\) if
\(u_m\le v_m\) for every coordinate \(m\).  For a finite design set
\(T_I\subset[0,1]^d\), define the lower-orthant acceptance set
\[
  \Gamma(T_I)=\{z\in[0,1]^d:\exists i\in I\text{ such that }z\preceq z_i\}.
\]
A full-sample point \(z_i\) is maximal if no other sample \(z_j\) satisfies
\(z_i\preceq z_j\) with \(z_j\ne z_i\).  Let \(B_n(T_n)\) be the set of
maximal indices.  Assume no duplicate points, which holds almost surely under a
continuous distribution.

\begin{proposition}[Pareto-frontier boundary]\label{prop:pareto}
The lower-orthant rule with maximal-index boundary is a proper projective
boundary scheme.
\end{proposition}

\begin{proof}
Fix \(T_n\) and a split \(I,J\).  If all omitted points are accepted by
\(\Gamma(T_I)\) and a maximal point \(b\) were omitted, then \(z_b\preceq z_i\)
for some \(i\in I\), contradicting maximality.  Hence all maximal points are in
\(I\).  Conversely, if all maximal points are in \(I\), every finite partially
ordered set element is dominated by a maximal element, so every omitted point is
accepted.  This proves boundary equivalence.  Projectivity follows because,
after all maximal points are retained, any retained non-maximal point is still
dominated by one of them; no new maximal point can appear.
\end{proof}

For conformal prediction, a calibration example \((X_i,Y_i)\) can be mapped to
a vector score \(R_i=\phi(X_i,Y_i)\).  A candidate label \(y\) for a new
covariate \(x\) is accepted when \(\phi(x,y)\in\Gamma(T_I)\).  Proposition
\ref{prop:pareto} says that the nondominated calibration scores are exactly the
samples needed to reconstruct the acceptance set.

For \(d=2\) and i.i.d. uniform scores on \([0,1]^2\), the profile is explicit.
Sort the sample by the first coordinate.  The ranks of the second coordinates
form a uniform random permutation, and Pareto maxima are the right-to-left
records of that permutation.  Therefore
\begin{equation}\label{eq:stirling-profile}
  p_n(k)=\frac{c(n,k)}{n!},
\end{equation}
where \(c(n,k)\) is the unsigned Stirling number of the first kind.  Equivalently,
\(c(1,1)=1\) and
\[
  c(n+1,k)=n\,c(n,k)+c(n,k-1).
\]
Combining this profile with Theorem~\ref{thm:exact-profile} gives
\begin{equation}\label{eq:pareto-moments}
  \Ee[(1-V_N)^M\mid K_N=k]
  =
  \frac{c(N+M,k)/(N+M)!}{c(N,k)/N!}
  \frac{\binom{N}{k}}{\binom{N+M}{k}}.
\end{equation}
In particular,
\begin{equation}\label{eq:first-moment}
  \Ee[V_N\mid K_N=k]
  =
  1-
  \frac{p_{N+1}(k)}{p_N(k)}
  \frac{N+1-k}{N+1}.
\end{equation}
The beta mean \(k/(N+1)\) is recovered only when the profile ratio
\(p_{N+1}(k)/p_N(k)\) equals one.

\section{PAC Certificates from the Exact Law}\label{sec:certificates}

The exact law gives a conditional PAC certificate by inverting the conditional
distribution.

\begin{definition}[Profile-based conditional quantile]
Assume \(p_N(k)>0\), and let \(\mu_{N,k}\) be the law in
Corollary~\ref{cor:law-from-profile}.  For \(\beta\in(0,1)\), define
\[
  q_{N,k}(\beta)
  :=
  \inf\{q\in[0,1]:\mu_{N,k}([0,q])\ge1-\beta\}.
\]
\end{definition}

\begin{corollary}[Sharp conditional PAC certificate]\label{cor:conditional-pac}
Under the assumptions of Theorem~\ref{thm:exact-profile},
\[
  \Pp\{V_N\le q_{N,k}(\beta)\mid K_N=k\}\ge1-\beta .
\]
If \(\mu_{N,k}\) has no atom at \(q_{N,k}(\beta)\), equality holds.
\end{corollary}

When the profile is not known exactly, one may work with a certified family of
profiles.

\begin{definition}[Robust profile class]
Let \(\cP\) be a family of admissible profiles.  If \(q_{N,k}^{(p)}(\beta)\)
is the quantile produced by profile \(p\in\cP\), define
\[
  q_{N,k}^{\rm rob}(\beta):=
  \sup_{p\in\cP} q_{N,k}^{(p)}(\beta).
\]
\end{definition}

\begin{corollary}[Robust conditional certificate]
If the true complexity profile belongs to \(\cP\), then
\[
  \Pp\{V_N\le q_{N,k}^{\rm rob}(\beta)\mid K_N=k\}\ge1-\beta .
\]
\end{corollary}

This is the precise role of prior or auxiliary information about the
complexity profile.  Without such information, Section~\ref{sec:no-go} shows
that conditioning on the observed value \(K_N=k\) alone cannot give a
nontrivial distribution-free guarantee.

\subsection{Computing or Bounding the Profile in Practice}\label{subsec:profile-computation}
Corollary~\ref{cor:law-from-profile} separates the universal probabilistic part of the certificate from a rule-specific statistical input, given by the complexity profile.  This should be read constructively:  the profile is the additional object that must be supplied, bounded, or estimated in order to obtain sharp conditional certification for random-boundary rules.  When no such information is available, the no-go result in Section~\ref{sec:no-go} explains why the observed value \(K_N=k\) alone cannot support a nontrivial distribution-free conditional statement.

There are three practically distinct regimes.

\begin{enumerate}
\item \emph{Analytic profiles.}  In simple projective schemes the profile can be computed exactly.  Scalar order-statistic calibration has deterministic boundary size, fixed-support scenario programs have stable support dimension under the usual nondegeneracy assumptions, and the Pareto-frontier example in Section~\ref{subsec:pareto} has the explicit record profile \(p_n(k)=c(n,k)/n!\) in dimension two.  In such cases the quantile in Corollary~\ref{cor:conditional-pac} is an exact finite-sample certificate with no simulation step.

\item \emph{Structural profile classes.}  In more complex procedures, exact formulas may be unavailable but deterministic properties can still restrict the admissible profiles.  Examples include upper bounds on the boundary size, monotonicity inherited from a recursive construction, decompositions into independent or nested boundary components, or envelopes \(\underline p_n(k)\le p_n(k)\le\overline p_n(k)\).  These constraints define a family \(\cP\) of admissible profiles, and the robust quantile \(q_{N,k}^{\rm rob}(\beta)\) converts that partial structural information into a conservative conditional PAC certificate.

\item \emph{Simulation-certified profiles.}  When the design distribution is known, or when a validated simulator is part of the statistical model, the profile can be estimated directly.  For each pair \((n,k)\) used by the moment or quantile calculation, run the boundary algorithm on \(R\) independent samples of size \(n\), and set
\[
  \widehat p_n(k)=\frac1R\sum_{r=1}^R
  \one\{K_n^{(r)}=k\}.
\]
For any finite grid \(\mathcal G\) of such pairs, the simultaneous binomial band
\[
  \max_{(n,k)\in\mathcal G}|\widehat p_n(k)-p_n(k)|
  \le
  \sqrt{\frac{\log(2|\mathcal G|/\delta)}{2R}}
\]
holds with probability at least \(1-\delta\) over the independent simulation runs.  Intersecting these bands with the elementary constraints \(p_n(k)\ge0\), \(\sum_k p_n(k)=1\), and the Hausdorff moment admissibility constraints implicit in \eqref{eq:profile-moments} gives a certified profile class \(\cP\).  With simulation probability at least \(1-\delta\), this class contains the true profile and the robust quantile gives the conditional PAC guarantee in Corollary~\ref{cor:conditional-pac}.  If the simulator is approximate or bootstrap-based, the profile calculation is model-based and inherits the simulator approximation.
\end{enumerate}

This profile step identifies the precise statistical information needed for conditional certification.  The beta law is recovered when this information reduces to profile stability; when the boundary is random, the profile factor is the finite-sample correction that prevents overconfident certificates.

\section{About Discarded Samples}\label{sec:discarding}

Theorem~\ref{thm:exact-profile} allows discarded samples, but only when they are true boundary
samples.  We next discuss this important distinction.

\subsection{Classical Violated-Discard Bounds}
In the standard sample-and-discard setting for scenario optimization, one removes \(r\) sampled constraints and returns a solution that typically violates them.  The results of Calafiore and of Campi--Garatti cover broad data-dependent removal mechanisms; see \cite{CalafioreSIOPT2010,CampiGaratti2011}.  The resulting risk bounds contain an additional combinatorial prefactor.  Violation of the removed constraints provides a weaker property than boundary equivalence: it does not ensure that those constraints reconstruct the decision, nor that accepted non-boundary samples can be deleted while preserving the removal path.  The prefactor reflects the wider class of admissible removal rules.

\subsection{Structured Support Removal}
More structured schemes, such as cascaded support removal, can have sharper
certificates.  In fully supported convex programs, repeatedly removing support
constraints gives a reproducible boundary candidate consisting of the removed
support constraints and the final support set, see
\cite{RomaoPapachristodoulouMargellos2023,
RomaoMargellosPapachristodoulou2023}.  When this boundary is essential and
projective, our Theorem~\ref{thm:exact-profile} recovers the no-prefactor beta-type law for the
certified acceptance set.  

Appendix~\ref{app:cascaded-support} gives a detailed verification of the cascaded support-removal procedure in \cite{RomaoPapachristodoulouMargellos2023}.  In the fully supported case, with \(r=\ell d\) discarded constraints, the union of the \(\ell\) removed support sets and the final support set has cardinality \((\ell+1)d=r+d\).  Proposition~\ref{prop:cascaded-support-boundary} shows that this union is a proper projective boundary for the certified acceptance set used in the compression argument.  If \(V_N^{\rm cs}\) and \(V_N^{\rm fin}\) denote the risks of the certified set and of the final optimizer, respectively, the following relations hold for every \(N>r+d\):
\[
  \Pp\{V_N^{\rm cs}>\eps\}
  =
  \sum_{i=0}^{r+d-1}\binom{N}{i}\eps^i(1-\eps)^{N-i},
  \qquad
  V_N^{\rm fin}\le V_N^{\rm cs},
\]
under the non-atomicity condition stated in the appendix.  This yields the no-prefactor feasibility guarantee for the final optimizer.  Under the additional tightness assumption of \cite[Theorem~5]{RomaoPapachristodoulouMargellos2023}, the certified set and the final feasible set differ only on a finite zero-probability set, and the same tail formula holds with equality for \(V_N^{\rm fin}\).

The extension in \cite{RomaoMargellosPapachristodoulou2023} accommodates an arbitrary number of discarded constraints through additional bookkeeping.  Its bound can be conservative when the number of discards is not an integer multiple of \(d\), so we do not treat it as an exact fixed-boundary special case here.

\subsection{Order-Statistic Discarding}
Scalar conformal prediction with \(r\) discarded upper-tail scores is a clean
projective case.  The \(r\) discarded scores and the threshold score form a
fixed boundary of size \(r+1\), and the exact law is
\[
  \Pp\{V_N>\eps\}
  =
  \sum_{i=0}^{r}\binom{N}{i}\eps^i(1-\eps)^{N-i}.
\]

The following sufficient condition summarizes the operational meaning of
essential discards.

\begin{proposition}[Essential Discards]\label{prop:essential-discards}
Let an algorithm produce a candidate boundary \(B_n=C_n\cup D_n\), where
\(C_n\) are retained support samples and \(D_n\) are discarded or exception
samples.  Interpret the certified object as including only the information that
is used to define the certified acceptance set, together with any exception
structure or deterministic removal path that is explicitly part of the
certificate.  Suppose that, for every deterministic data set under consideration:
\begin{enumerate}
\item the certified object, and hence the certified acceptance set, can be
reconstructed from \(B_n\);
\item every sample outside \(B_n\) is accepted by the certified acceptance set;
\item no element of \(B_n\) can be removed without changing this certified
object;
\item adding or deleting accepted non-boundary samples does not change the
certified object or the reported boundary.
\end{enumerate}
Then \(B_n\) is a proper projective boundary, and
Theorem~\ref{thm:exact-profile} applies with \(K_n=|B_n|\).
\end{proposition}

\begin{proof}
The first item is boundary reconstruction, the second is outside-boundary
feasibility, and the third is the minimality condition in
Lemma~\ref{lem:primitive}, with the certified object replacing the raw optimizer
or removal path whenever those objects are part of the certificate.  The fourth
item gives confirmed-addition stability and boundary projectivity.  Hence
Assumptions~\ref{ass:boundary-equivalence} and~\ref{ass:boundary-projectivity}
hold.
\end{proof}

We finally distinguish a mathematical boundary from the set of indices returned
by a particular implementation.  We call such an implementation-level output a
\emph{reported set} and denote it by \(\widetilde B_n(S_n)\).  The reported set
is the subset of training samples that the implementation stores, displays, or
uses as its certificate representation.  It need not itself satisfy
Assumptions~\ref{ass:boundary-equivalence} and~\ref{ass:boundary-projectivity}.
In particular, it may strictly contain a genuine proper projective boundary
\(B_n^\star(S_n)\).

This distinction matters because Theorem~\ref{thm:exact-profile} applies to a
genuine boundary, not to an arbitrary reported set.  If
\(\widetilde B_n(S_n)\) contains superfluous samples, then the implication
encoded in \eqref{eq:boundary-equivalence} can fail for \(\widetilde B_n\):
a superfluous reported point may be held out while all held-out samples are
still accepted.  Thus the event
\(\widetilde B_{N+M}(S_{N+M})\subseteq I\) may be too strong to characterize
the event that all held-out samples are accepted.  The safe use of an
over-reported set is instead as an upper bound on the size of some truly
proper projective boundary.

\begin{proposition}[Over-Reported Boundaries]\label{prop:superfluous}
Let \(s\ge 1\), and let \((\Gamma,B^\star)\) be a proper projective boundary
scheme for the certified acceptance set.  Suppose that, for every \(n\ge s\),
the implementation reports a set \(\widetilde B_n\) such that
\[
  B_n^\star(S_n)\subseteq \widetilde B_n(S_n),
  \qquad
  |\widetilde B_n(S_n)|\le s
  \quad\text{almost surely}.
\]
Then, for every \(N\ge s\) and every \(\epsilon\in(0,1)\),
\[
  \Pp\{V_N>\epsilon\}
  \le
  \sum_{i=0}^{s-1}\binom{N}{i}\epsilon^i(1-\epsilon)^{N-i}.
\]
\end{proposition}

\begin{proof}
Let \(K_n^\star=|B_n^\star(S_n)|\).  Since
\(B_n^\star(S_n)\subseteq\widetilde B_n(S_n)\) and
\(|\widetilde B_n(S_n)|\le s\), we have \(K_n^\star\le s\) almost surely for
every \(n\ge s\).  Corollary~\ref{cor:bounded_boundary}, applied to the truly
proper projective boundary scheme \((\Gamma,B^\star)\), gives the displayed
bound.
\end{proof}

Over-reporting can therefore be used conservatively through a genuine boundary-size upper bound.  A more serious failure occurs when violated or otherwise nonprojective discards are treated as boundary samples without first identifying a proper boundary contained in the reported set.  In that case neither Theorem~\ref{thm:exact-profile} nor Corollary~\ref{cor:bounded_boundary} applies; one should use specialized discarding bounds, the stable-compression result in Section~\ref{sec:compression}, or the inner-certificate construction of Proposition~\ref{prop:inner}.

\section{A Stable-Compression Fallback}\label{sec:compression}

Boundary equivalence is stronger than ordinary stability in the sense of
sample compression.  Some algorithms admit a stable compressed representation, that is,
after the reported compression, discard, or exception samples are retained,
all other samples can be deleted without changing the reconstructed decision.
Such algorithms need not satisfy boundary equivalence, because the retained
compressed set may reconstruct the decision without characterizing exactly
which held-out samples would be accepted or would enter the full-sample
boundary.  In these cases one can use a conservative sample-compression bound.

Let \(A(S_N)=a_N\) be the returned decision and let
\(\ell(a,z)\in\{0,1\}\) be the violation loss.  Define
\[
  V_N=\Pp\{\ell(a_N,Z)=1\mid S_N\},
  \qquad
  \widehat V_N=\frac1N\sum_{i=1}^N \ell(a_N,Z_i).
\]
Suppose the algorithm reports a compression set \(C_N\) and possibly a discard
or exception set \(D_N\).  Let
\[
  \kappa(S_N)=S_{C_N\cup D_N},
  \qquad
  b_N=|C_N\cup D_N|.
\]
A reconstruction map \(\rho\) satisfies \(a_N=\rho(\kappa(S_N))\).

\begin{definition}[Stable compressed decision]
The pair \((\kappa,\rho)\) is stable if deleting samples outside
\(C_N\cup D_N\) does not change the reconstructed decision.  Equivalently, for
every \(I\subseteq[N]\) containing \(C_N\cup D_N\), running the scheme on
\(S_I\) reconstructs the same decision as running it on \(S_N\).
\end{definition}

\begin{theorem}[Stable-compression PAC bound]\label{thm:compression}
Assume \((\kappa,\rho)\) is a stable compression scheme.  For every
distribution \(P\), every \(N\ge1\), and every \(\beta\in(0,1)\), with
probability at least \(1-\beta\),
\begin{equation}\label{eq:hk-bound}
  |V_N-\widehat V_N|
  \le
  \sqrt{
    \widehat V_N\,\frac{72}{N}
    \left(2b_N+\log\frac{4e}{\beta}\right)}
  +
  \frac{32}{N}
  \left(2b_N+\log\frac{4e}{\beta}\right).
\end{equation}
In particular, if at most \(r_N\) observed samples have loss one, then
\begin{equation}\label{eq:discard-bound}
  V_N
  \le
  \frac{r_N}{N}
  +
  \sqrt{
    \frac{72r_N}{N^2}
    \left(2b_N+\log\frac{4e}{\beta}\right)}
  +
  \frac{32}{N}
  \left(2b_N+\log\frac{4e}{\beta}\right).
\end{equation}
\end{theorem}

\begin{proof}
Equation \eqref{eq:hk-bound} is Corollary~18 of 
\cite{HannekeKontorovich2021}, applied to the binary loss
\(\ell(a,z)\).  If at most \(r_N\) samples have loss one, then
\(\widehat V_N\le r_N/N\), and the right side of
\eqref{eq:hk-bound} is nondecreasing in \(\widehat V_N\).
\end{proof}

This fallback is usually looser than the exact profile law.  Its value is that
it remains valid when the reported discards are stable compression elements but
not proper boundary elements.

\section{A Multi-Risk Extension}\label{sec:multi-risk}

We next discuss a mixed-moment law for several risks.  Let
\(h=1,\ldots,H\) index risk components.  Component \(h\) has sample size
\(N_h\), distribution \(P_h\), acceptance set \(\Gamma_h(S)\), and boundary
\(B_h(S)\).  The blocks are independent, and samples within each block are
i.i.d.  Define
\[
  V_h(S)=P_h\{Z_h\notin\Gamma_h(S)\mid S\},
  \qquad K_h=|B_h(S)|.
\]
For vectors
\[
  \mathbf N=(N_1,\ldots,N_H),\qquad
  \mathbf k=(k_1,\ldots,k_H),
\]
write
\[
  p_{\mathbf n}(\mathbf k)=\Pp\{(K_1,\ldots,K_H)=\mathbf k\}
\]
when the block sizes are \(\mathbf n\).

\begin{assumption}[Multi-risk boundary condition]\label{ass:multi}
For every deterministic collection of augmented blocks and every split into
design subsets \(I_h\) and held-out subsets \(J_h\),
\[
  T_{h,J_h}\subseteq\Gamma_h(T_I)\quad\text{for every }h
  \quad\Longleftrightarrow\quad
  B_h(T)\subseteq I_h\quad\text{for every }h.
\]
On this event, \(B_h(T_I)=B_h(T)\) for every \(h\), after re-indexing.
\end{assumption}

\begin{theorem}[Multi-risk profile law]\label{thm:multi}
Under Assumption~\ref{ass:multi}, for every nonnegative integer vector
\(\mathbf M=(M_1,\ldots,M_H)\),
\begin{equation}\label{eq:multi-unconditional}
  \Ee\!\left[
    \prod_{h=1}^H(1-V_h)^{M_h}
    \one_{\{\mathbf K_{\mathbf N}=\mathbf k\}}
  \right]
  =
  p_{\mathbf N+\mathbf M}(\mathbf k)
  \prod_{h=1}^H
  \frac{\binom{N_h}{k_h}}{\binom{N_h+M_h}{k_h}}.
\end{equation}
If \(p_{\mathbf N}(\mathbf k)>0\), then
\begin{equation}\label{eq:multi-conditional}
  \Ee\!\left[
    \prod_{h=1}^H(1-V_h)^{M_h}
    \mid \mathbf K_{\mathbf N}=\mathbf k
  \right]
  =
  \frac{p_{\mathbf N+\mathbf M}(\mathbf k)}{p_{\mathbf N}(\mathbf k)}
  \prod_{h=1}^H
  \frac{\binom{N_h}{k_h}}{\binom{N_h+M_h}{k_h}}.
\end{equation}
\end{theorem}

\begin{proof}
Add \(M_h\) fresh test samples to block \(h\).  Conditional on the training data,
the probability that all fresh samples are accepted in all components is
\(\prod_h(1-V_h)^{M_h}\).  Exchangeability within each block lets us replace the
fresh-sample experiment by a uniformly chosen design subset of size \(N_h\) in
each augmented block.  Assumption~\ref{ass:multi} then says that all held-out
samples are accepted exactly when every design subset contains its component
boundary, and projectivity identifies the restricted complexity vector with the
augmented one.  Conditional on the augmented data and on
\(\mathbf K_{\mathbf N+\mathbf M}=\mathbf k\), the probability that each design
subset contains its fixed \(k_h\)-point boundary is the product in
\eqref{eq:multi-unconditional}.  Taking expectations and then conditioning
proves the result.
\end{proof}

\begin{example}[Dedicated reserve sizing with several service risks]\label{ex:multi-reserve}
A concrete case satisfying Assumption~\ref{ass:multi} is a reserve-sizing
problem with \(H\) dedicated service regions.  Region \(h\) has demand
samples \(D_{h,1},\ldots,D_{h,N_h}\) drawn from \(P_h\), and a planner chooses
nonnegative reserve capacities by
\[
  \min_{x_1,\ldots,x_H}\sum_{h=1}^H c_h x_h
  \quad\text{subject to}\quad
  D_{h,i}\le x_h,
  \qquad i=1,\ldots,N_h,\quad h=1,\ldots,H .
\]
Thus \(x_h(S)=\max_{i\le N_h}D_{h,i}\),
\(\Gamma_h(S)=\{d:d\le x_h(S)\}\), and \(V_h\) is the conditional
probability of a shortage in region \(h\).  Let \(B_h(S)\) be the
index of the demand sample attaining the maximum in block \(h\), with the
fixed tie-breaking convention of Section~\ref{sec:setup}.

For any augmented collection of blocks and any design subsets \(I_h\), all
held-out demands are accepted in every region if and only if, in every block,
the augmented block maximum is retained in \(I_h\).  This is exactly
\(B_h(T)\subseteq I_h\) for all \(h\).  Once these maxima are retained,
deleting accepted nonmaximal demands cannot change any \(x_h\), so the
projectivity part of Assumption~\ref{ass:multi} also holds.  Under continuous
demand distributions \(K_h=1\) for all \(h\), and Theorem~\ref{thm:multi}
reduces to the joint mixed moments of independent beta laws for the shortage
probabilities.  The example also shows the limitation of the assumption: if a
single shared reserve could be reallocated across regions, then retaining each
component's individual maximum would generally no longer be equivalent to
retaining the joint boundary of the coupled decision.
\end{example}

If
\[
  p_{\mathbf N+\mathbf M}(\mathbf k)=p_{\mathbf N}(\mathbf k)
  \qquad\text{for all }\mathbf M\ge0,
\]
then the conditional mixed moments factor into beta mixed moments, and the
components \(V_h\) are conditionally independent with
\(V_h\sim\Beta(k_h,N_h-k_h+1)\), with a point mass at zero when \(k_h=0\).
When the profile is not stable, the ratio
\(p_{\mathbf N+\mathbf M}(\mathbf k)/p_{\mathbf N}(\mathbf k)\) carries the
dependence information.  This gives a route to non-Bonferroni certificates when
a joint or componentwise projective boundary is available.

A common-sample joint-risk version is obtained by applying the single-risk
theorem to the joint acceptance set \(\Gamma_\cap(S)=\cap_h\Gamma_h(S)\).  If
that joint set has a proper boundary \(B_\cap\), then the risk
\[
  V_\cup(S)=\Pp\{Z\notin\Gamma_\cap(S)\mid S\}
\]
satisfies the same moment law with the profile of \(B_\cap\).

\section{No-Go Results}\label{sec:no-go}

The next counterexamples show why a universal beta theorem cannot be based only on the event \(K_N=k\).

\begin{proposition}[A nominal boundary is not enough]\label{prop:circle-gap}
There is a permutation-invariant, empirically consistent algorithm with a
natural two-point nominal boundary for which the \(s=2\) beta upper bound fails.
\end{proposition}

\begin{proof}
Let \(\cZ\) be the unit circle with the uniform distribution.  Given
\(N\ge2\) sample points, let \(G_N\) be the largest open circular gap between
consecutive sample points, with deterministic tie-breaking, and set
$
  \Gamma(S_N)=\cZ\setminus G_N$.
All samples are accepted.  The violation probability is the length of the
largest gap, \(V_N=|G_N|\).  The two endpoints of \(G_N\) are a natural nominal
boundary.  But the \(N\) gaps sum to one, so the largest gap is strictly larger
than \(1/N\) almost surely.  Hence
\[
  \Pp\{V_N>1/N\}=1.
\]
The beta tail with \(s=2\) at \(\eps=1/N\) equals
\[
  (1-1/N)^N+N(1/N)(1-1/N)^{N-1}
  =
  \left(2-\frac1N\right)\left(1-\frac1N\right)^{N-1}<1.
\]
Thus the beta upper bound fails.  Boundary projectivity fails because adding an
accepted point can split the largest gap and change the decision.
\end{proof}

\begin{theorem}[No conditional theorem from complexity alone]\label{thm:no-go}
Fix \(N\ge1\), \(k\in\{1,\ldots,N\}\), and \(q<1\).  There exists a
permutation-invariant, empirically consistent, sample-compressed algorithm such
that \(\Pp\{K_N=k\}>0\) but
\[
  \Pp\{V_N>q\mid K_N=k\}=1 .
\]
Consequently, any distribution-free conditional guarantee of the form
\[
  \Pp\{V_N\le u_N(k,\beta)\mid K_N=k\}\ge1-\beta
\]
valid for all such algorithms must have \(u_N(k,\beta)=1\) for every
\(\beta<1\).
\end{theorem}

\begin{proof}
Take \(\cZ=[0,1]\) with the uniform distribution.  Choose \(\delta,\eta>0\)
such that \(\delta+k\eta<1-q\).  Let
$  E_\delta=\{\max_i Z_i\le\delta\}$.
This event has probability \(\delta^N>0\).

On \(E_\delta^c\), set \(\Gamma(S_N)=[0,1]\) and report \(K_N=0\).  On
\(E_\delta\), let \(Z_{(1)}<\cdots<Z_{(k)}\) be the \(k\) smallest order
statistics.  Choose \(k\) disjoint slots in \((\delta,1)\), each of length
larger than \(\eta\), and in slot \(j\) place an interval
\(I_j(Z_{(j)})\) of length \(\eta\), with its left endpoint depending
injectively and measurably on \(Z_{(j)}\).  Define
$
  \Gamma(S_N)=[0,\delta]\cup\bigcup_{j=1}^k I_j(Z_{(j)})$.
All observed samples lie in \([0,\delta]\), so the rule is empirically
consistent.  On \(E_\delta\),
\[
  P(\Gamma(S_N))=\delta+k\eta,
  \qquad
  V_N=1-\delta-k\eta>q .
\]
The \(k\) interval locations encode \(Z_{(1)},\ldots,Z_{(k)}\), so the reported
decision is reconstructed from those \(k\) samples.  The construction is
permutation invariant because it uses order statistics.  Therefore
\(\{K_N=k\}=E_\delta\) up to null sets, and the conditional statement follows.
\end{proof}

This theorem is consistent with the conditional-risk obstruction discussed in
\cite{GarattiCampi2023Conditional}.  A nontrivial conditional certificate based on \(K_N=k\) therefore requires structural assumptions, profile information, or both.

\section{Numerical Illustrations}\label{sec:numerics}

Numerical experiments are provided to illustrate how the exact profile factor changes
finite-sample risk statements when the observed boundary size is random.

\subsection{A Random-Support Scenario Envelope} \label{subsec:numerics-envelope}

For the coordinatewise envelope in Section~\ref{subsec:coordinatewise-envelope},
\(1-V_N\) is the product of the two coordinatewise sample maxima.  Under the
uniform distribution on \([0,1]^2\), these two maxima are independent
\(\Beta(N,1)\) variables.  Therefore, for \(0\le\varepsilon<1\),
\begin{equation}\label{eq:coordinate-envelope-tail}
  \mathbb P\{V_N>\varepsilon\mid K_N=k\}
  =(1-\varepsilon)^N\{1-N\log(1-\varepsilon)\},
  \qquad k=1,2.
\end{equation}
This formula agrees with the moment sequence
\eqref{eq:coordinate-envelope-moments}.  It also makes clear why inserting the
observed random support size into a fixed-dimension beta formula is not valid.
For \(N=20\), the conditional mean is
\[
  \mathbb E[V_{20}\mid K_{20}=k]
  =1-\left(\frac{20}{21}\right)^2
  =0.09297,
  \qquad k=1,2.
\]
The fixed-dimension beta means would be \(1/21=0.04762\) for \(k=1\) and
\(2/21=0.09524\) for \(k=2\), see Figure~\ref{fig:coordinatewise-tail}.

\begin{figure}[tbp]
  \centering
  \includegraphics[width=.65\linewidth]{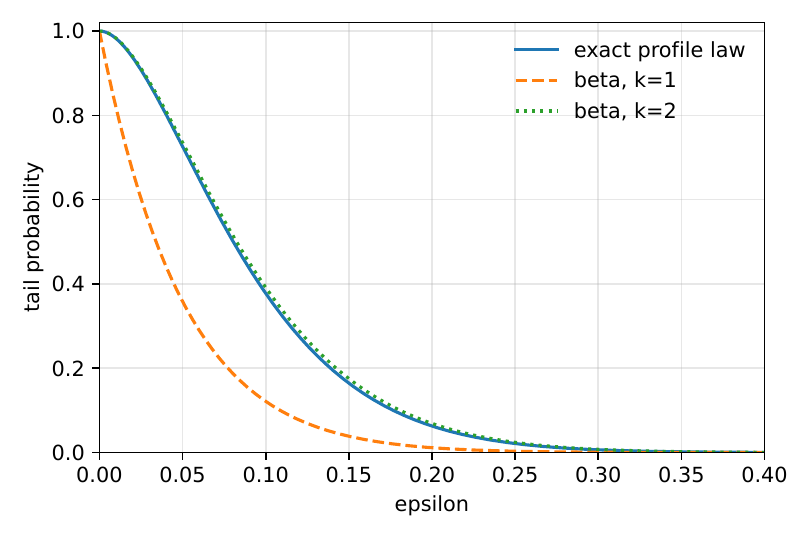}
  \caption{Exact conditional tail for the coordinatewise scenario envelope with
  \(N=20\), compared with the two beta tails that would be obtained by treating
  the random support size as fixed.  The same projective-profile law holds on
  both events \(K_N=1\) and \(K_N=2\).}
  \label{fig:coordinatewise-tail}
\end{figure}

\subsection{Pareto-Frontier Vector-Score Calibration}\label{subsec:numerics-pareto}

We next return to the Pareto-frontier rule of
Proposition~\ref{prop:pareto}.  Figure~\ref{fig:pareto-region} shows one
calibration sample and the corresponding lower-orthant acceptance boundary.  A
new score vector is rejected exactly when it lies above the staircase, in which
case it would become a new maximal sample after augmentation.

\begin{figure}[tbp]
  \centering
  \includegraphics[width=.60\linewidth]{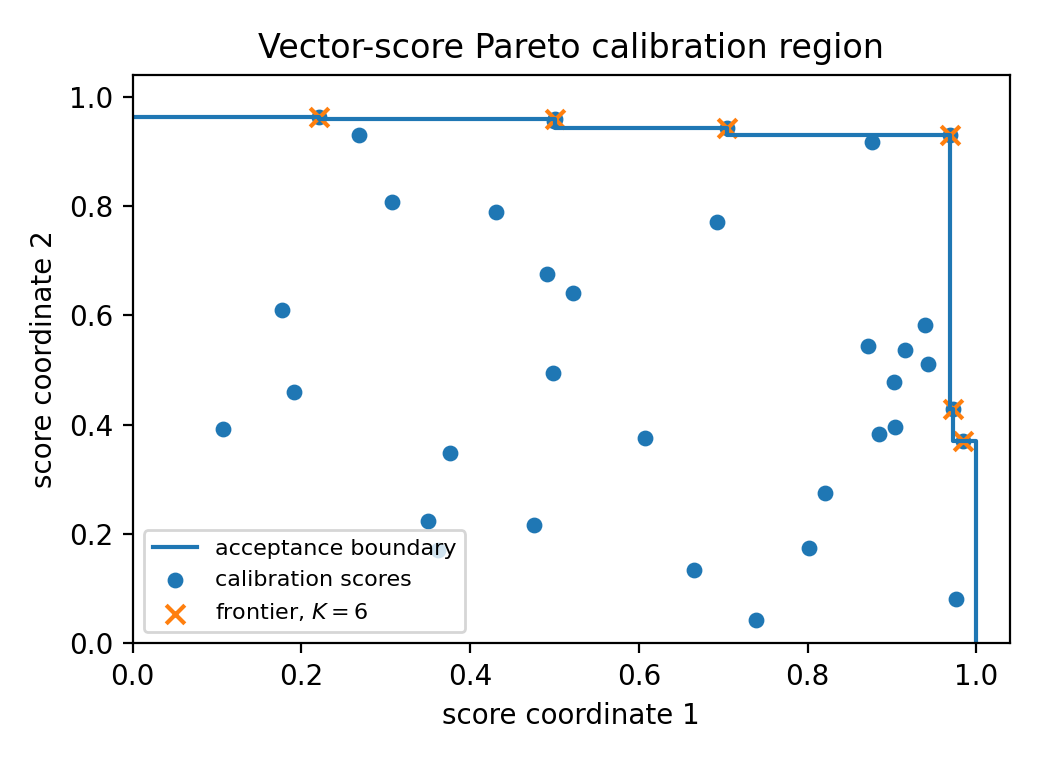}
  \caption{A bivariate score sample and its Pareto-frontier acceptance boundary.
  The marked points are the nondominated calibration scores, i.e., the boundary
  samples in Proposition~\ref{prop:pareto}.}
  \label{fig:pareto-region}
\end{figure}

For i.i.d. uniform scores in \([0,1]^2\), the frontier-size profile is the record
profile \(p_n(k)=c(n,k)/n!\), where \(c(n,k)\) is the unsigned Stirling number of
the first kind.  Table~\ref{tab:pareto-means} and Figure~\ref{fig:pareto-profile-mean} compare the exact conditional mean from
\eqref{eq:first-moment} with the beta mean \(k/(N+1)\) for \(N=20\).  The beta
mean is smaller for small observed frontiers and larger for large observed
frontiers; neither direction is uniformly safe.

\begin{table}[tbp]
\centering
\begin{tabular}{@{}rrrr@{}}
\toprule
\(k\) & \(p_{21}(k)/p_{20}(k)\) & exact mean & beta mean \(k/21\) \\
\midrule
1 & 0.95238 & 0.09297 & 0.04762 \\
2 & 0.96580 & 0.12618 & 0.09524 \\
3 & 0.98312 & 0.15733 & 0.14286 \\
4 & 1.00457 & 0.18678 & 0.19048 \\
5 & 1.03061 & 0.21477 & 0.23810 \\
6 & 1.06193 & 0.24148 & 0.28571 \\
7 & 1.09947 & 0.26702 & 0.33333 \\
8 & 1.14450 & 0.29150 & 0.38095 \\
9 & 1.19874 & 0.31501 & 0.42857 \\
\bottomrule
\end{tabular}
\caption{Exact conditional means for the two-dimensional Pareto-frontier rule
with \(N=20\).}
\label{tab:pareto-means}
\end{table}

\begin{figure}[tbp]
  \centering
  \includegraphics[width=.72\linewidth]{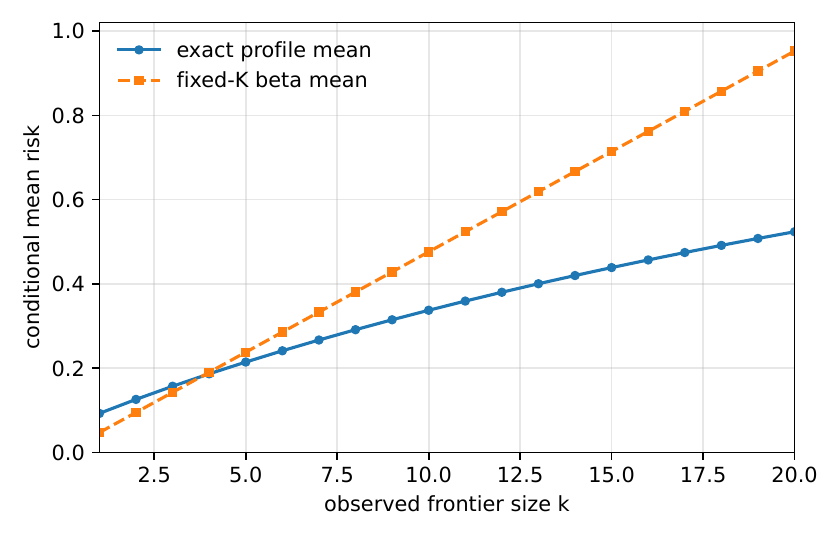}
  \caption{Exact profile means and fixed-\(K\) beta means for the
  two-dimensional Pareto-frontier rule with \(N=20\).}
  \label{fig:pareto-profile-mean}
\end{figure}

For \(k=3\), the exact first-moment calculation gives
\[
  p_{20}(3)=0.2748198358,
  \qquad
  \frac{p_{21}(3)}{p_{20}(3)}=0.9831174498,
\]
and therefore
\[
  \mathbb E[V_{20}\mid K_{20}=3]=0.1573279001.
\]
The beta law with fixed complexity \(3\) would instead give \(3/21=0.1428571429\).
Figure~\ref{fig:pareto-cdf} and Table~\ref{tab:pareto-conditional-sim} compare the empirical conditional distribution,
obtained by direct Monte Carlo conditioning on \(K_{20}=3\), with the incorrect
\(\Beta(3,18)\) comparator.

\begin{table}[tbp]
\centering
\begin{tabular}{@{}rrrr@{}}
\toprule
Monte Carlo mean & Monte Carlo 95\% quantile & beta 95\% quantile & conditional draws \\
\midrule
0.157189 & 0.295490 & 0.282619 & 160000\\
\bottomrule
\end{tabular}
\caption{Conditional simulation for the Pareto-frontier rule with \(N=20\) and
\(K_N=3\).  The beta quantile is included only as a comparator; it is not the
projective-profile law.}
\label{tab:pareto-conditional-sim}
\end{table}

\begin{figure}[tbp]
  \centering
  \includegraphics[width=.72\linewidth]{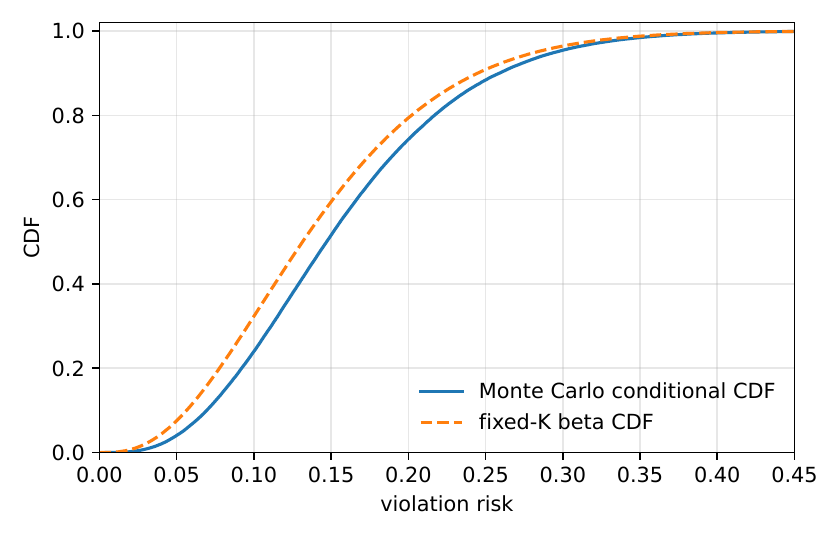}
  \caption{Conditional CDF of the Pareto-frontier violation risk for \(N=20\)
  and \(K_N=3\), estimated by Monte Carlo, compared with the fixed-boundary beta
  law.  The beta curve is visibly shifted to the left and gives a smaller
  95\% quantile.}
  \label{fig:pareto-cdf}
\end{figure}

The numerical evidence confirms that a beta law is exact only when the complexity
profile is stable, for instance in scalar order-statistic calibration or fixed
support-dimension scenario programs.  In random-boundary problems, the observed
value \(K_N=k\) must be interpreted together with the cross-sample profile.

\section{Conclusions}\label{sec:conclusion}

This paper gives an exact finite-sample risk law for proper projective boundaries.  The result identifies a deterministic mechanism behind beta-binomial laws in scenario optimization, split conformal prediction, and related distribution-free certification methods: boundary equivalence and projectivity reduce the risk law to the cross-sample complexity profile.

The classical beta law appears when the boundary size is fixed, or more generally when the complexity profile is stable at the observed value.  In random-boundary settings, including coordinatewise scenario envelopes and Pareto-frontier calibration with vector scores, the profile factor is generally unavoidable.  Treating the observed random boundary size as if it were a fixed dimension can therefore give incorrect conditional risk assessments.

For algorithms with discarded samples, the result gives a diagnostic rather than a replacement for the specialized sample-and-discard theory.  Essential projective discards may be counted in the boundary and certified by the sharp law.  Generic violated discards should instead be handled by established discarding bounds, by a certified inner acceptance set, or by conservative stable-compression inequalities.

From the viewpoint of statistical predictive inference, the main practical message is that sharp conditional certification requires a complexity profile.  This profile may be analytic, structurally bounded, or estimated from a simulator with simultaneous finite-sample bands.  Without such information, the no-go theorem shows that the observed complexity alone is insufficient for nontrivial distribution-free conditional guarantees.

Several directions remain open: tighter structural bounds for admissible profiles, certified quantile computation from finitely many moments, broader classes of multivariate conformal rules with projective boundaries, and multi-risk profile methods that avoid Bonferroni allocation when joint boundary information is available.

\paragraph{Data Availability.}
No external datasets were used.  The numerical values in the figures and tables can be reproduced from the formulas and Monte Carlo procedures described in the paper.

\appendix

\section{Proof of the Bounded-Boundary Corollary}
\label{app:padded-proof}

\begin{proof}[Proof of Corollary~\ref{cor:bounded_boundary}]

The case \(s=0\) gives \(K_n=0\) almost surely for every \(n\ge1\), and is therefore
covered by Corollary~\ref{cor:fixed-size}.  We henceforth assume \(s\ge1\).

The idea is to enlarge the boundary by auxiliary non-boundary points until its
size is exactly \(s\), and then to apply the fixed-size result to a conservative
acceptance rule.  Let
$  \widetilde{\cZ}:=\cZ\times[0,1]$,
and let
$
  \widetilde Z_i=(Z_i,U_i)$,
where the \(U_i\)'s are i.i.d. uniform on \([0,1]\), independent of the
\(Z_i\)'s.  The auxiliary marks are distinct with probability one.  The argument is carried out on this full-measure event; any measurable convention may be used on its complement.

For a finite extended sample
$
  \widetilde T_n=((z_1,u_1),\ldots,(z_n,u_n))$,
write \(T_n=(z_1,\ldots,z_n)\), \(B_n=B_n(T_n)\), and
\(b_n=|B_n|\).  Define the padded boundary map as follows.  If \(n<s\), set
$
  \widetilde B_n(\widetilde T_n):=[n]$.
If \(n\ge s\) and \(b_n\le s\), let \(R_n(\widetilde T_n)\) be the set of the
\(s-b_n\) largest auxiliary marks among the indices in \([n]\setminus B_n\), and
set
$
  \widetilde B_n(\widetilde T_n)
  :=
  B_n(T_n)\cup R_n(\widetilde T_n)$.
If \(n\ge s\) and \(b_n>s\), set
$
  \widetilde B_n(\widetilde T_n):=B_n(T_n)$.
The last case is irrelevant almost surely under the assumptions of the
corollary, but it keeps the rule defined on every deterministic sample.

We now define the associated padded acceptance rule.  For a finite index set
\(I\), write \(T_I\) and \(\widetilde T_I\) for the corresponding subsamples,
and let
$
  B_I:=B_{|I|}(T_I)$, $ b_I:=|B_I|$.
If \(|I|<s\), set
$
  \widetilde\Gamma(\widetilde T_I):=\varnothing$.
If \(|I|\ge s\) and \(b_I>s\), set
$
  \widetilde\Gamma(\widetilde T_I):=\Gamma(T_I)\times[0,1]$.
Finally, if \(|I|\ge s\) and \(b_I\le s\), let
$
  r_I:=s-b_I$.
When \(r_I=0\), set \(\tau_I=1\).  When \(r_I\ge1\), let \(\tau_I\) be the
\(r_I\)-th largest auxiliary mark among the indices in \(I\setminus B_I\).
This is well-defined because \(|I|\ge s\).  Define
\[
  \widetilde\Gamma(\widetilde T_I)
  :=
  \{(z,u): z\in\Gamma(T_I),\ u\le \tau_I\}.
\]
On the full-measure event relevant to the corollary, the padded rule is a conservative version of the original rule.

We verify boundary equivalence and projectivity for the padded scheme.  Fix a
deterministic extended sample \(\widetilde T_n\) and a split
\(I\subseteq[n]\), \(J=[n]\setminus I\).

First suppose \(n<s\).  Then \(\widetilde B_n(\widetilde T_n)=[n]\), so
\(\widetilde B_n(\widetilde T_n)\subseteq I\) holds if and only if \(I=[n]\).
Since \(|I|<s\), we have \(\widetilde\Gamma(\widetilde T_I)=\varnothing\), and
all held-out points are accepted if and only if \(J=\varnothing\), again
equivalently \(I=[n]\).  Boundary equivalence follows, and projectivity is
trivial in this case.

Now assume \(n\ge s\).  We first consider the case \(b_n>s\).  Then
$
  \widetilde B_n(\widetilde T_n)=B_n(T_n)$.
If all held-out extended points are accepted by
\(\widetilde\Gamma(\widetilde T_I)\), then in particular all their
\(z\)-components are accepted by \(\Gamma(T_I)\).  Boundary equivalence for the
original scheme gives \(B_n(T_n)\subseteq I\), hence
\(\widetilde B_n(\widetilde T_n)\subseteq I\).  Conversely, if
\(\widetilde B_n(\widetilde T_n)\subseteq I\), then \(B_n(T_n)\subseteq I\).
Original projectivity gives
\[
  B_{|I|}(T_I)=B_n(T_n),
\]
after the natural re-indexing, and hence \(b_I=b_n>s\).  Therefore
$
  \widetilde\Gamma(\widetilde T_I)=\Gamma(T_I)\times[0,1]$.
Original boundary equivalence gives \(z_j\in\Gamma(T_I)\) for all \(j\in J\),
so all held-out extended points are accepted.  This proves boundary equivalence
when \(b_n>s\).  Projectivity in this case follows immediately from original
projectivity, since whenever \(\widetilde B_n(\widetilde T_n)\subseteq I\),
\[
  \widetilde B_{|I|}(\widetilde T_I)
  =
  B_{|I|}(T_I)
  =
  B_n(T_n)
  =
  \widetilde B_n(\widetilde T_n).
\]

It remains to consider the case \(n\ge s\) and \(b_n\le s\).  Put
\[
  C:=B_n(T_n), \qquad A:=[n]\setminus C, \qquad r:=s-|C|,
\]
and let \(R:=R_n(\widetilde T_n)\) be the set of the \(r\) largest auxiliary
marks among the indices in \(A\).  Thus
$
  \widetilde B_n(\widetilde T_n)=C\cup R$.
Suppose first that
$
  \widetilde B_n(\widetilde T_n)\subseteq I$.
Then \(C\subseteq I\), so original boundary equivalence gives
$
  z_j\in\Gamma(T_I)$, $j\in J$.
Original projectivity also gives
$
  B_{|I|}(T_I)=C$.
Since \(R\subseteq I\), the \(r\) largest auxiliary marks among the full
non-boundary set \(A\) are all retained in the design set.  Hence every omitted
non-boundary point \(j\in A\setminus I\) has auxiliary mark at most the
threshold \(\tau_I\).  Therefore
$
  (z_j,u_j)\in\widetilde\Gamma(\widetilde T_I)$,
 $ j\in J$.

Conversely, suppose that all held-out extended points are accepted by
\(\widetilde\Gamma(\widetilde T_I)\).  Then all their \(z\)-components are
accepted by \(\Gamma(T_I)\).  By original boundary equivalence,
$
  C=B_n(T_n)\subseteq I$.
By original projectivity,
$
  B_{|I|}(T_I)=C$.
Thus the threshold \(\tau_I\) is the \(r\)-th largest auxiliary mark among the
retained non-boundary indices \(I\setminus C\), with the convention
\(\tau_I=1\) when \(r=0\).  Since the auxiliary marks are distinct on a
full-measure event,
\[
  R\subseteq I
  \quad\Longleftrightarrow\quad
  u_j\le \tau_I
  \ \text{for every } j\in A\setminus I .
\]
Indeed, if one of the \(r\) largest full-sample non-boundary marks were omitted,
then the \(r\)-th largest retained non-boundary mark would be strictly smaller
than the omitted mark.  Since all held-out extended points are accepted, the
right-hand condition holds, and hence \(R\subseteq I\).  Therefore
$
  \widetilde B_n(\widetilde T_n)=C\cup R\subseteq I
$.
This proves boundary equivalence for the padded scheme.

The same argument gives projectivity.  If
$
  \widetilde B_n(\widetilde T_n)=C\cup R\subseteq I$,
then original projectivity gives \(B_{|I|}(T_I)=C\).  Moreover, because all
indices in \(R\) are retained, they remain exactly the \(r=s-|C|\) largest
auxiliary marks among the retained non-boundary indices.  Hence, after the
natural re-indexing,
$
  \widetilde B_{|I|}(\widetilde T_I)
  =
  C\cup R
  =
  \widetilde B_n(\widetilde T_n)
$.
Thus \((\widetilde\Gamma,\widetilde B)\) is a proper projective boundary scheme on the full-measure regularity class described above.

By assumption, for every \(n\ge s\),
$
  K_n\le s
  \qquad\text{almost surely}.
$
Therefore the padded boundary satisfies
$
  |\widetilde B_n(\widetilde S_n)|=s$
 almost surely for every \(n\ge s\).
Applying Corollary~\ref{cor:fixed-size} to the padded scheme gives, for
\(N\ge s\),
$
  \widetilde V_N\sim \Beta(s,N-s+1)$,
where
$
  \widetilde V_N
  :=
  \widetilde{\Pp}\{(Z,U)\notin
  \widetilde\Gamma(\widetilde S_N)\mid \widetilde S_N\}
$
and \(\widetilde{\Pp}\) denotes probability under the product law
\(P\otimes{\rm Unif}[0,1]\).

Finally, for every training sample in the full-measure regularity class with \(N\ge s\),
$
  \widetilde\Gamma(\widetilde S_N)
  \subseteq
  \Gamma(S_N)\times[0,1]$.
Consequently the padded rule is pointwise more conservative, and therefore
\[
\begin{aligned}
  V_N
  &=
  \widetilde{\Pp}\{(Z,U)\notin
  \Gamma(S_N)\times[0,1]\mid \widetilde S_N\}  \\
  &\le
  \widetilde{\Pp}\{(Z,U)\notin
  \widetilde\Gamma(\widetilde S_N)\mid \widetilde S_N\}
  =
  \widetilde V_N .
\end{aligned}
\]
Hence, for every \(\epsilon\in(0,1)\),
\[
  \Pp\{V_N>\epsilon\}
  = \widetilde{\Pp}\{V_N>\epsilon\}
  \le
  \widetilde{\Pp}\{\widetilde V_N>\epsilon\}
  =
  \sum_{i=0}^{s-1}
  \binom{N}{i}\epsilon^i(1-\epsilon)^{N-i},
\]
where the last equality is the standard beta-binomial tail identity for
\(\Beta(s,N-s+1)\).  This proves the claim.
\end{proof}

\section{Cascaded Support Removal as a Projective Boundary}\label{app:cascaded-support}
This appendix places the cascaded support-removal scheme of \cite{RomaoPapachristodoulouMargellos2023} within the projective-boundary framework.  The construction used in the compression proof becomes a proper projective boundary for a certified acceptance set.  The corresponding no-prefactor formula then follows from the fixed-boundary law, while the final optimizer inherits the resulting upper bound.

\subsection{The Cascade}

Consider a convex scenario program
\[
  \widehat x(H)\in\argmin_{x\in X} c^\top x
  \quad\text{subject to}\quad
  g(x,\delta)\le0,\qquad \delta\in H,
\]
where \(H\) is a finite set of scenarios.  As in \cite{RomaoPapachristodoulouMargellos2023}, assume feasibility and a uniquely selected optimizer for every finite scenario set under consideration.  Ordinary uniqueness or a fixed lexicographic rule can provide this selection.  The statements below are understood on the joint full-measure regularity class on which these properties and the fully supported condition hold for every subproblem used by the cascade.

Fix an integer \(\ell\ge0\), let \(d\) be the decision dimension, set
\(r=\ell d\), and set
\[
  \zeta=(\ell+1)d=r+d .
\]
For a sample \(S\) with \(|S|>\zeta\), define the cascade as follows.  Start from
\(S^{(0)}=S\).  At stage \(k=0,\ldots,\ell\), solve the scenario program on
\(S^{(k)}\) and write
$
  \widehat x_k(S):=\widehat x(S^{(k)})$.
Let \(R_k(S)\) be the support set of this stage, namely the scenarios in
\(S^{(k)}\) whose removal changes the selected optimizer.  For
\(k<\ell\), remove these support constraints and set
\[
  S^{(k+1)}=S^{(k)}\setminus R_k(S).
\]
The final set \(R_\ell(S)\) is the support set of the final problem; it is not
removed.  In the fully supported case treated in
\cite[Theorem~3]{RomaoPapachristodoulouMargellos2023}, each \(R_k(S)\) has
cardinality \(d\).  The cascaded support set is
\begin{equation}\label{eq:cascade-boundary}
  C(S):=\bigcup_{k=0}^{\ell} R_k(S),
  \qquad |C(S)|=\zeta.
\end{equation}
This union is the boundary candidate.

For any candidate set \(C\) of cardinality \(\zeta\), run the same cascade on
\(C\).  Define
\[
  A_1(C):=\{\delta:g(\widehat x_\ell(C),\delta)\le0\}
\]
and
\[
  A_3(C):=\bigcup_{k=0}^{\ell-1} R_k(C),
\]
with \(A_3(C)=\emptyset\) when \(\ell=0\).  The set \(A_1(C)\) is the feasible
set of the final optimizer.  The set \(A_3(C)\) adds back the finitely many scenarios removed along the cascade.

The compression proof of \cite[Theorem~3]{RomaoPapachristodoulouMargellos2023}
uses one more set.  For \(k=0,\ldots,\ell\), let the union over previous removed
support sets be empty when \(k=0\), and define
\[
  A_2(C):=\bigcap_{k=0}^{\ell}
  \bigcap_{\substack{J\subset C\setminus\cup_{h=0}^{k-1}R_h(C)\\ |J|=d-1}}
  \{\delta:c^\top \widehat x(J\cup\{\delta\})
        \le c^\top \widehat x_k(C)\}.
\]
Then set
\begin{equation}\label{eq:cascade-certified-set}
  A^{\rm cs}(C):=(A_1(C)\cap A_2(C))\cup A_3(C).
\end{equation}
This is the certified acceptance set in the compression argument.  It may be
smaller than \(A_1(C)\), so its violation probability can be larger than the
violation probability of the final optimizer.  This is why it gives a valid
upper bound for the final optimizer.

For the exact tightness result of
\cite[Theorem~5]{RomaoPapachristodoulouMargellos2023}, an additional assumption
is imposed.  In the notation above, it says that if a scenario \(\delta\) is a
support constraint at any stage \(k\), then \(\delta\) is violated by every
optimizer that could be obtained from any \(d\)-point subset of the remaining
constraints after deleting \(\delta\).  Under that assumption, the set
\(A_2(C)\) is no longer needed and one can use
\begin{equation}\label{eq:cascade-tight-set}
  \overline A(C):=A_1(C)\cup A_3(C).
\end{equation}
Under a non-atomic scenario law, the finite set \(A_3(C)\) has probability zero.  The violation risk of \(\overline A(C)\) then equals the violation risk of the final optimizer \(\widehat x_\ell(C)\).

\subsection{A Compression-to-Projectivity Lemma}

We first state a general deterministic fact which bridges the
compression terminology of \cite{RomaoPapachristodoulouMargellos2023} and the
projective-boundary terminology of this paper.

\begin{lemma}[Unique compression implies projectivity]\label{lem:unique-compression-projective}
Fix \(\zeta\ge1\).  Suppose that, for every finite data set \(T\) with
\(|T|\ge\zeta\), a rule selects a subset \(C(T)\subseteq T\) with
\(|C(T)|=\zeta\).  Suppose also that an acceptance map \(G(C)\) is defined for
every \(\zeta\)-point set \(C\), and that the following two properties hold for
every \(T\) with \(|T|\ge\zeta\):
\begin{enumerate}
\item \emph{consistency:} every point of \(T\) is accepted by \(G(C(T))\);
\item \emph{unique compression:} if \(C\subseteq T\), \(|C|=\zeta\), and every
point of \(T\) is accepted by \(G(C)\), then \(C=C(T)\).
\end{enumerate}
Define \(\Gamma(T)=G(C(T))\) for \(|T|\ge\zeta\), define
\(\Gamma(T)=\emptyset\) for \(|T|<\zeta\), and define
\(B(T)=C(T)\) for \(|T|\ge\zeta\).  Then, for every full sample size
\(n\ge\zeta\), \((\Gamma,B)\) is a proper projective boundary scheme on samples
of size \(n\), and \(|B(T)|=\zeta\).
\end{lemma}

\begin{proof}
Fix a deterministic full sample \(T_n\), \(n\ge\zeta\), and a split
\(I\subseteq[n]\), \(J=[n]\setminus I\).  Write \(T_I\) for the retained sample.
If \(|I|<\zeta\), then \(B(T_n)\subseteq I\) is impossible, while
\(\Gamma(T_I)=\emptyset\) and \(J\neq\emptyset\), so the left side of boundary
equivalence is also false.  Hence assume \(|I|\ge\zeta\).

First suppose every omitted point \(T_J\) is accepted by \(\Gamma(T_I)\).  By
consistency, every retained point \(T_I\) is also accepted by \(\Gamma(T_I)
=G(C(T_I))\).  Hence every point of the full sample \(T_n\) is accepted by
\(G(C(T_I))\).  Since \(C(T_I)\subseteq T_I\subseteq T_n\) and has cardinality
\(\zeta\), unique compression for the full sample gives
\(C(T_I)=C(T_n)\).  Therefore \(B(T_n)=C(T_n)\subseteq I\).

Conversely, suppose \(B(T_n)=C(T_n)\subseteq I\).  Consistency for the full
sample says that every point of \(T_n\), and therefore every point of \(T_I\),
is accepted by \(G(C(T_n))\).  Since \(C(T_n)\subseteq T_I\) and has cardinality
\(\zeta\), unique compression for the retained sample gives
\(C(T_I)=C(T_n)\).  Thus
\[
  \Gamma(T_I)=G(C(T_I))=G(C(T_n))=\Gamma(T_n).
\]
Consistency for the full sample now implies that every omitted point in \(T_J\)
is accepted by \(\Gamma(T_I)\).  This proves boundary equivalence.

The same argument also proves projectivity.  Whenever \(B(T_n)\subseteq I\), we
have just shown that \(C(T_I)=C(T_n)\), which is exactly
\(B(T_I)=B(T_n)\) after the natural re-indexing.
\end{proof}

\subsection{Application to the Romao--Papachristodoulou--Margellos Cascade}

The proof of \cite[Theorem~3]{RomaoPapachristodoulouMargellos2023} establishes that, under feasibility, unique selection, and the fully supported condition, the set
\(C(S)\) in \eqref{eq:cascade-boundary} is the unique compression set of
cardinality \(\zeta=(\ell+1)d\) for the map \(A^{\rm cs}\) in
\eqref{eq:cascade-certified-set}.  In the notation of
Lemma~\ref{lem:unique-compression-projective}, take
\[
  G(C)=A^{\rm cs}(C),
  \qquad B(S)=C(S).
\]
The lemma gives the following conclusion.

\begin{proposition}[Cascaded support boundary]\label{prop:cascaded-support-boundary}
Under the feasibility, unique-selection, and fully supported conditions of \cite[Theorem~3]{RomaoPapachristodoulouMargellos2023}, the cascaded support set \(C(S)\) is, on the regularity class described above, a proper projective boundary of fixed size \(\zeta=r+d\) for the certified acceptance set \(A^{\rm cs}(C(S))\).
Consequently, for every \(N>r+d\), if \(V_N^{\rm cs}\) is the violation risk of this certified set,
then
\[
  V_N^{\rm cs}\sim \Beta(r+d,N-r-d+1)
\]
and
\begin{equation}\label{eq:cascade-beta-bound}
  \Pp\{V_N^{\rm cs}>\eps\}
  =
  \sum_{i=0}^{r+d-1}\binom{N}{i}\eps^i(1-\eps)^{N-i}.
\end{equation}
If the scenario law is non-atomic, the final optimizer's violation risk \(V_N^{\rm fin}\) satisfies
\[
  V_N^{\rm fin}\le V_N^{\rm cs}.
\]
Hence
\eqref{eq:cascade-beta-bound} recovers the feasibility bound of
\cite[Theorem~3]{RomaoPapachristodoulouMargellos2023}.
\end{proposition}

\begin{proof}
The cited unique-compression result and Lemma~\ref{lem:unique-compression-projective} give the proper-projective claim on the full-measure regularity class.  The boundary size is fixed and equal to \(r+d\), so Corollary~\ref{cor:fixed-size} gives the beta law and the beta-binomial tail
\eqref{eq:cascade-beta-bound}.

Finally, by construction,
\[
  A^{\rm cs}(C)=(A_1(C)\cap A_2(C))\cup A_3(C)\subseteq A_1(C)\cup A_3(C).
\]
Under a non-atomic scenario law, the finite set \(A_3(C)\) has probability zero.  The certified set consequently has no larger acceptance probability than the final optimizer's feasible set, and its violation risk is no smaller.
\end{proof}

Under the additional tightness assumption used in
\cite[Theorem~5]{RomaoPapachristodoulouMargellos2023}, the same reasoning
applies to the larger set \(\overline A\) in \eqref{eq:cascade-tight-set}.
Indeed, \cite[Theorem~5]{RomaoPapachristodoulouMargellos2023} proves that
\(C(S)\) is the unique compression set of cardinality \(\zeta\) for
\(\overline A\).  Lemma~\ref{lem:unique-compression-projective} then gives a
proper projective boundary for \(\overline A(C(S))\).  Under a non-atomic scenario law, \(\overline A(C(S))\) differs from the feasible set of the final optimizer only by the zero-probability set \(A_3(C(S))\).  Its risk therefore equals the usual scenario violation risk, and, for every \(N>r+d\), Corollary~\ref{cor:fixed-size} gives
\[
  \Pp\{V_N^{\rm fin}>\eps\}
  =
  \sum_{i=0}^{r+d-1}\binom{N}{i}\eps^i(1-\eps)^{N-i},
\]
which is the equality statement in
\cite[Theorem~5]{RomaoPapachristodoulouMargellos2023}.

\section{A Deterministic Verification Lemma}\label{app:primitive}

The following lemma packages a useful set of sufficient conditions for boundary equivalence.

\begin{lemma}[Primitive conditions for boundary equivalence]\label{lem:primitive}
Fix a deterministic data set \(T_n\).  Suppose a boundary map \(B_n(T_n)\)
satisfies:
\begin{enumerate}
\item \emph{confirmed-addition stability:} if \(I\subseteq K\subseteq[n]\) and
\(z_j\in\Gamma(T_I)\) for all \(j\in K\setminus I\), then
\(\Gamma(T_K)=\Gamma(T_I)\);
\item \emph{boundary reconstruction:}
\(\Gamma(T_{B_n})=\Gamma(T_{[n]})\);
\item \emph{outside-boundary feasibility:}
\(z_j\in\Gamma(T_{[n]})\) for every \(j\notin B_n(T_n)\);
\item \emph{minimality:} if
\(\Gamma(T_{[n]\setminus\{j\}})=\Gamma(T_{[n]})\), then
\(j\notin B_n(T_n)\).
\end{enumerate}
Then boundary equivalence \eqref{eq:boundary-equivalence} holds.
\end{lemma}

\begin{proof}
Let \(I\subseteq[n]\) and \(J=[n]\setminus I\).  First suppose all
\(z_j\), \(j\in J\), are accepted by \(\Gamma(T_I)\).  Confirmed-addition
stability gives \(\Gamma(T_{[n]})=\Gamma(T_I)\).  For any \(j\in J\), adding
the points in \(J\setminus\{j\}\) also leaves the decision unchanged, so
\(\Gamma(T_{[n]\setminus\{j\}})=\Gamma(T_{[n]})\).  Minimality implies
\(j\notin B_n(T_n)\).  Thus \(B_n(T_n)\subseteq I\).

Conversely, suppose \(B_n(T_n)\subseteq I\).  By reconstruction,
\(\Gamma(T_{B_n})=\Gamma(T_{[n]})\).  Every point in \(I\setminus B_n\) is
outside the boundary and is accepted by \(\Gamma(T_{[n]})\).  Repeated
confirmed-addition stability gives \(\Gamma(T_I)=\Gamma(T_{[n]})\).  If
\(j\in J\), then \(j\notin B_n\), so outside-boundary feasibility gives
\(z_j\in\Gamma(T_{[n]})=\Gamma(T_I)\).
\end{proof}

\section{Inner Certificates}\label{app:inner}

An algorithm that falls outside the boundary-equivalence framework may still admit a proper inner certificate.

\begin{proposition}[Proper inner certificate]\label{prop:inner}
Fix \(1\le s\le N\).  Suppose the algorithm returns \(\Gamma(S_N)\), and suppose there is a certified inner set \(\Gamma^-(S_N)\subseteq\Gamma(S_N)\) satisfying the fixed-size proper boundary assumptions with boundary size \(s\).  Let \(V_N\) and \(V_N^-\) be
the violation risks of \(\Gamma(S_N)\) and \(\Gamma^-(S_N)\).  Then
\[
  \Pp\{V_N>\eps\}
  \le
  \sum_{i=0}^{s-1}\binom{N}{i}\eps^i(1-\eps)^{N-i}.
\]
\end{proposition}

\begin{proof}
Since \(\Gamma^-(S_N)\subseteq\Gamma(S_N)\), \(V_N\le V_N^-\) pointwise.  The
fixed-boundary law gives \(V_N^-\sim\Beta(s,N-s+1)\), and the beta tail gives
the claim.
\end{proof}

\bibliographystyle{plainnat}
\bibliography{biblio}

\end{document}